\documentclass[12pt]{article}

\usepackage{amssymb}
\usepackage{amsmath}
\usepackage{amsthm}
\usepackage[utf8]{inputenc}
\usepackage[english]{babel}
\usepackage{graphicx}
\usepackage{color}

\usepackage[a4paper,margin=2.5cm]{geometry}

\theoremstyle{plain}
\newtheorem{theorem}{Theorem}[section]
\newtheorem{proposition}[theorem]{Proposition}
\newtheorem{lemma}[theorem]{Lemma}

\theoremstyle{definition}

\newtheorem{assumption}[theorem]{Assumption}

\numberwithin{equation}{section}

\newcommand{\abs}[1]{\lvert{#1}\rvert}

\newcommand{\ket}[1]{\lvert{#1}\rangle}
\newcommand{\bra}[1]{\langle{#1}\rvert} 
\newcommand{\ip}[2]{\langle{#1},{#2}\rangle}

\DeclareMathOperator{\im}{Im}
\DeclareMathOperator{\re}{Re} 

\DeclareMathOperator{\Ker}{Ker}

\DeclareMathOperator{\Ai}{Ai}
\DeclareMathOperator{\supp}{supp}

\newcommand{\bR}{\mathbf{R}}
\newcommand{\bC}{\mathbf{C}}

\newcommand{\cF}{\mathcal{F}}

\newcommand{\cH}{\mathcal{H}}

\newcommand{\cG}{\mathcal{G}}
\newcommand{\cB}{\mathcal{B}}
\newcommand{\cK}{\mathcal{K}}
\newcommand{\cR}{\mathcal{R}}
\newcommand{\cS}{\mathcal{S}}

\newcommand{\sfr}{\mathsf{r}}
\newcommand{\sfs}{\mathsf{s}}
\newcommand{\sfu}{\mathsf{u}}
\newcommand{\sfv}{\mathsf{v}}

\newcommand{\eps}{\varepsilon}

\begin{document}
\title{Instability of resonances under Stark perturbations}
\author{Arne Jensen\footnote{Department of Mathematical Sciences, Aalborg University, Skjernvej 4A, DK-9220 Aalborg \O{}, Denmark. E-mail: \texttt{matarne@math.aau.dk}}
\and
Kenji Yajima\footnote{Department of Mathematics, Gakushuin University, 1-5-1 Mejiro, Toshima-ku, Tokyo 171-8588, Japan. E-mail: \texttt{kenji.yajima@gakushuin.ac.jp}}}
\date{}

\maketitle
\begin{abstract}
Let $H^{\eps}=-\frac{d^2}{dx^2}+\eps x +V$, $\eps\geq0$, on $L^2(\bR)$. Let 
$V=\sum_{k=1}^Nc_k\ket{\psi_k}\bra{\psi_k}$ be a rank $N$ operator, where the $\psi_k\in L^2(\bR)$ are real, compactly supported, and even. Resonances are defined using analytic scattering theory. The main result is that if $\zeta_n$, $\im\zeta_n<0$, are resonances of $H^{\eps_n}$ for a sequence $\eps_n\downarrow0$ as $n\to\infty$ and $\zeta_n\to\zeta_0$ as $n\to\infty$, $\im\zeta_0<0$, then $\zeta_0$ is \emph{not} a resonance of $H^0$.
\end{abstract}


\section{Introduction}
We consider a family of Hamiltonians on $\cH=L^2(\bR)$ given as
\begin{equation}\label{model}
H_0^{\eps}=-\frac{d^2}{dx^2}+\eps x, \quad H^{\eps}=H_0^{\eps}+V, 
\quad \eps\in[0,\infty),
\end{equation}
where $V$ is a bounded self-adjoint operator. Under suitable assumptions on $V$ one can define resonances of $H^{\eps}$ as poles of matrix elements $\ip{u}{(H^{\eps}-\zeta)^{-1}v}$ continued analytically from the upper halfplane across $(0,\infty)$ to the lower halfplane. Consider the following situation. Suppose that 
$H^0$ has a resonance $\zeta_0$ in the lower halfplane close to the positive real axis. Suppose that
there exists a sequence $\eps_n\downarrow0$ for $n\to\infty$ such that each $H^{\eps_n}$ has resonance $\zeta_n$ in the lower half-plane. We then ask: Is it possible that $\zeta_n\to\zeta_0$ as $n\to\infty$? The main result here is that under suitable conditions on $V$ the answer is \textbf{no}. One example is that $V$ is a rank one operator $V=c\ket{\psi}\bra{\psi}$ such that $\psi\in L^2(\bR)$ has compact support and is a real-valued function.

The instability of pre-existing resonances was first considered in \cite{HR}. They obtained results for two different models, a Friedrich model, and a model of the form 
\eqref{model} with $V$ a rank one perturbation. An explicit construction of a dilation analytic rank one perturbation leading to a resonance of $H^0$ close to the real axis is given.  Then as $\eps\downarrow0$ all resonances of $H^{\eps}$ are converging to the real axis, i.e. do not converge to a pre-existing resonance, see~\cite[Theorem 1.13]{HR}. Their proofs rely of detailed studies of the resolvent behavior.

We obtain results for a class of perturbations different from the one in \cite{HR}. We use techniques from abstract analytic scattering theory. Stationary scattering theory for Stark Hamiltonians was first obtained in~\cite{Y79} and results on analytic scattering theory for Stark Hamiltonians was obtained in~\cite{Y81}. An abstract analytic scattering theory was given in~\cite{AJ}. In particular the identity between poles of the analytically continued matrix elements of the resolvent and poles of the analytically continued scattering matrix was shown. This result was obtained in~\cite{Y81} for Stark Hamiltonians.

Our main results are stated in Theorem~\ref{rank-one} for rank one perturbations, and in Theorem~\ref{rankN} for a rank $N$ perturbation under the assumptions that it is given by compactly supported, real, and even functions. The proofs rely on the connection between poles of analytically continued matrix elements of the resolvent and poles of the analytically continued scattering matrix, and a detailed analysis of asymptotics of the Airy function.

\section{Notation and framework}\label{sect2}
We consider the following families of Hamiltonians on $\cH=L^2(\bR)$:
\begin{equation}
H_0^{\eps}=-\frac{d^2}{dx^2}+\eps x, \quad H^{\eps}=H_0^{\eps}+V, 
\quad \eps\in[0,\infty).
\end{equation}

The perturbation $V$ is assumed to be a bounded self-adjoint operator on $\cH$ which is factored as $V=B^{\ast}A=A^{\ast}B$. Here $A,B\colon \cH\to\cK$ are bounded operators and $\cK$ is an auxiliary Hilbert space. Further assumptions on $A$ and $B$ will be stated later.

We start by recalling a variant of the notation used in the Kuroda approach to scattering theory~\cite{Kuroda}. We define
\begin{equation}
R_0^{\eps}(\zeta)=(H_0^{\eps}-\zeta)^{-1},\quad R^{\eps}(\zeta)=(H^{\eps}-\zeta)^{-1},
\quad \im\zeta\neq0.
\end{equation}
We also define
\begin{equation}\label{def-G}
Q_0^{\eps}(\zeta)=BR_0^{\eps}(\zeta)A^{\ast},\quad G_0^{\eps}(\zeta)=1+Q_0^{\eps}(\zeta).
\end{equation}

We have that $G_0^{\eps}(\zeta)$ is invertible for $\im\zeta\neq0$. The second resolvent equation can be written as
\begin{equation}\label{resolvent-eq}
R^{\eps}(\zeta)=R_0^{\eps}(\zeta)-R_0^{\eps}(\zeta)A^{\ast}G_0^{\eps}(\zeta)^{-1}B
R_0^{\eps}(\zeta).
\end{equation}

We need the spectral representation for $H_0^{\eps}$. Since the spectral multiplicity is $2$ for $\eps=0$ and $1$ for $\eps>0$, we split into these two cases. For $\eps=0$ the spectral representation is $F^0\colon L^2(\bR)\to L^2((0,\infty);\bC^2)$ defined as
\begin{equation}
(F^0u)(\lambda)=T^0(\lambda)u=\begin{bmatrix}
T_0^+(\lambda)u
\\
T_0^-(\lambda)u
\end{bmatrix}=\frac{1}{\sqrt{2}\lambda^{1/4}}\begin{bmatrix}
\widehat{u}(\sqrt{\lambda})\\
\widehat{u}(-\sqrt{\lambda})
\end{bmatrix},\quad \lambda>0.
\end{equation}
The operator of multiplication by $\lambda$ on $L^2((0,\infty);\bC^2)$ is denoted by $M_{\lambda}$. Then  $F^0$ is unitary and we have $F^0H_0^0(F^0)^{\ast}=M_{\lambda}$.

For the case $\eps>0$ we define
\begin{align}
(V(\eps)u)(x)&=\frac{1}{\sqrt{\eps}}u(\frac{1}{\eps}x),\\
U(\eps)&=V(\eps)
\cF^{\ast}M_{\exp(-ip^3/(3\eps))}\cF.\label{U-def}
\end{align}
Then $F^{\eps}$  given by $(F^{\eps}u)(\lambda)=(U(\eps)u)(\lambda)$ is unitary, and we have $F^{\eps}H_0^{\eps}(F^{\eps})^{\ast}=M_{\lambda}$, see \cite{Y79}.

The trace operators used in the Kuroda approach are  defined as follows for $v\in\cK$
\begin{align}
T^{\eps}(\lambda;A)v&=(F^{\eps}A^{\ast}v)(\lambda),\label{Teps-def}\\
T^{\eps}(\lambda;B)v&=(F^{\eps}B^{\ast}v)(\lambda).
\end{align}

We will assume that there exists $\Omega\subseteq\bC$ satisfying $\overline{\Omega}=\Omega$ such that $\Omega\cap\bR=I$ is an open interval satisfying $I\subseteq(0,\infty)$. We assume that $T^{\eps}(\lambda;A)$ and $T^{\eps}(\lambda;B)$ have analytic extensions to $\Omega$ with values in $\cB(\cK,\bC^2)$ for $\eps=0$ and in 
$\cB(\cK,\bC)$ for $\eps>0$.

Let $\bC^{\pm}=\{\zeta\,|\,\pm\im\zeta>0\}$. We define
\begin{equation}
\Omega^{\pm}=\{\zeta\in\Omega\,|\,\pm\im\zeta>0\}.
\end{equation}
and then
\begin{equation}
Q^{\eps}_{0,\pm}(\zeta)=Q_0^{\eps}(\zeta),\quad \zeta\in\bC^{\pm}.
\end{equation}
We recall
\begin{proposition}[{\cite[Proposition 3.1]{AJ}}]\label{prop21}
We have the following results:
\begin{enumerate}
\item
$Q^{\eps}_{0,+}(\zeta)$ has an analytic continuation from $\bC^+$ to $\bC^+\cup I\cup\Omega^-$, which we denote by $\widetilde{Q}^{\eps}_{0,+}(\zeta)$.
\item
$Q^{\eps}_{0,-}(\zeta)$ has an analytic continuation from $\bC^-$ to $\bC^-\cup I\cup\Omega^+$, which we denote by $\widetilde{Q}^{\eps}_{0,-}(\zeta)$.
\item
We have  for $\zeta\in\Omega$
\begin{equation}
\widetilde{Q}^{\eps}_{0,+}(\zeta)-\widetilde{Q}^{\eps}_{0,-}(\zeta)
=2\pi i T^{\eps}(\overline{\zeta};B)^{\ast}T^{\eps}(\zeta;A).
\end{equation}
\end{enumerate}
\end{proposition}
We use the notation
\begin{equation}
\widetilde{G}_{0,\pm}^{\eps}(\zeta)=1+\widetilde{Q}^{\eps}_{0,\pm}(\zeta), 
\quad \zeta\in\Omega.
\end{equation}
We impose assumptions on $A$ and $B$ such that $\widetilde{Q}^{\eps}_{0,\pm}(\zeta)$ is compact for $\zeta\in\Omega$. We can then use the analytic Fredholm theorem to obtain the following result.

\begin{proposition}[{\cite[Proposition 3.2]{AJ}}]
There exist discrete sets $e_{\pm}^{\eps}\subset I$ with the end points of $I$ as the only possible points of accumulation, and discrete sets $r_{\pm}^{\eps}\subset\Omega^{\mp}$ with $\partial\Omega^{\mp}\setminus I$ as the only possible points of accumulation. Then $\widetilde{G}_{0,\pm}^{\eps}(\zeta)$ are invertible for $\zeta\in(\bC^{\pm}\cup
\Omega^{\mp}\cup I)\setminus(e_{\pm}^{\eps}\cup r_{\pm}^{\eps})$. The continued inverse    $(\widetilde{G}_{0,\pm}^{\eps}(\zeta))^{-1}$ has poles contained in the set $e_{\pm}^{\eps}\cup r_{\pm}^{\eps}$.
\end{proposition}

We define $\bC^{(\eps)}$, such that $\bC^{(0)}=\bC^2$ and $\bC^{(\eps)}=\bC$, $\eps>0$.
We introduce the dense subsets 
\begin{multline}
\cR_0^{\eps}=\{f\in L^2(I;\bC^{(\eps)})\,|\,
 f\colon I\to \bC^{(\eps)}\\ \text{has an analytic continuation to $\Omega$ with values in $\bC^{(\eps)}$}\}.
 \end{multline}
 
We then have the result that for $f,g\in (F^{\eps})^{-1}\cR_0^{\eps}$ the matrix element $\ip{f}{R_0^{\eps}(\zeta)g}$ has an analytic continuation from $\bC^{\pm}$ to $\bC^{\pm}\cup I\cup\Omega^{\mp}$, see \cite[Proposition 3.6]{AJ}. Using \eqref{resolvent-eq} we can get a meromorphic continuation of matrix elements of the full resolvent. For each $\eps\geq0$ we have that $e_+^{\eps}=e_-^{\eps}=e^{\eps}=I\cap\sigma_p(H^{\eps})$, see \cite[Theorem 3.9]{AJ}.

In the sequel we will only consider $r_+^{\eps}$. These points are the possible locations of poles of the meromorphically continued full resolvent matrix elements in the lower half plane, and are called the resonances. We now recall the results from \cite{AJ} identifying these with poles of the meromorphically continued scattering matrix.

For $\zeta\in(\bC^+\cup I\cup \Omega^{-})\setminus(e_{+}^{\eps}\cup r_{+}^{\eps})$ we introduce the notation $\widetilde{G}_+^{\eps}(\zeta)=\widetilde{G}_{0,+}^{\eps}(\zeta)^{-1}$. We define
$\widetilde{G}_-^{\eps}(\zeta)$ analogously. 

We have the following formulas for the scattering matrix and its inverse, see \cite[Theorem 3.11]{AJ}.
\begin{align}
S^{\eps}(\lambda)&=1-2\pi i T^{\eps}(\lambda;A) \widetilde{G}_+^{\eps}(\lambda)
T^{\eps}(\overline{\lambda};B)^{\ast},\\
S^{\eps}(\lambda)^{-1}&=1+2\pi i T^{\eps}(\lambda;A) \widetilde{G}_-^{\eps}(\lambda)
T^{\eps}(\overline{\lambda};B)^{\ast}.\label{Sinv}
\end{align}
We have a meromorphic extension of $S^{\eps}(\lambda)$ to $\Omega$ with poles at most in $r_+^{\eps}$. Note that the singularities in $e^{\eps}$ are removable. Analogously for $S^{\eps}(\lambda)^{-1}$, now with poles at most in $r_-^{\eps}$.

The main result in \cite{AJ} is the following theorem.
\begin{theorem}[{\cite[Theorem 3.12]{AJ}}]\label{main}
The set of poles of $S^{\eps}(\zeta)$ in $\Omega$ is equal to the set $r_+^{\eps}$. For a given $\eps\geq0$ and $\zeta_0\in r_+^{\eps}$ we have that $\Ker(\widetilde{G}_{0,+}^{\eps}(\zeta_0))$ is isomorphic to $\Ker(S^{\eps}(\zeta_0)^{-1})$.
\end{theorem}

The relation between existence of a resonance and existence of a non-zero solution to 
$S^{\eps}(\zeta_0)^{-1}u=0$ given by \eqref{Sinv} will be used to study the stability or instability of resonances for a sequence $\eps_n$ with $\eps_n\to0$ as $n\to\infty$.

\section{Rank one perturbation}\label{sect3}
We consider the case of $V$ a rank one perturbation. We assume $V=c\ket{\psi}\bra{\psi}$ for some vector $\psi\in L^2(\bR)$, $\psi\neq0$, and $c$ real, $c\neq0$. We take $\cK=\bC$ and $A=\bra{\psi}$, $B=c\bra{\psi}$.

Consider first the case $\eps=0$. We have for $z\in\cK$
\begin{equation}\label{T0}
T^0(\lambda;A)z=\frac{1}{\sqrt{2}\lambda^{1/4}}\begin{bmatrix}
\widehat{\psi}(\sqrt{\lambda})\\
\widehat{\psi}(-\sqrt{\lambda})
\end{bmatrix}z.
\end{equation}
The determination of  $\sqrt{\lambda}$ is the one with $\sqrt{\lambda}>0$ for $\lambda>0$ and the cut along $(-\infty,0]$.
We need to be able to continue this operator analytically in $\lambda$. We introduce the following assumption, where $L^2_{\rm comp}(\bR)$ denotes the compactly supported functions in $L^2$.

\begin{assumption}\label{assump31}
Assume $\psi\in L^2_{\rm comp}(\bR)$. 
\end{assumption}
It follows from this assumption that $\widehat{\psi}$ has an analytic continuation from $\bR$ to the complex plane $\bC$.

We take $\Omega=\bC\setminus(-\infty,0]$. Then it follows from \eqref{T0} and Assumption~\ref{assump31} that $T^0(\lambda;A)$ can be continued analytically to $\Omega$. We have the same result for $T^0(\lambda;B)=cT^0(\lambda;A)$.

We now consider $\zeta$ with $\re\zeta>0$ and $\im\zeta<0$. We continue \eqref{T0} to these $\zeta$. We also have 
\begin{equation}
T^0(\overline{\zeta},B)^{\ast}=
\frac{c}{\sqrt{2}\zeta^{1/4}}\begin{bmatrix}
\widehat{\overline{\psi}}(-\sqrt{\zeta}) & \widehat{\overline{\psi}}(\sqrt{\zeta})
\end{bmatrix},
\end{equation}since $\overline{\widehat{\psi}(\sqrt{\overline{\zeta}})}
=\widehat{\overline{\psi}}(-\sqrt{\zeta})$.

Continuing \eqref{Sinv} to $\{\zeta\,|\,\re\zeta>0,\; \im\zeta<0\}$ we 
get the following components of the matrix $S^0(\zeta)^{-1}$. 
\begin{align}
(S^0(\zeta)^{-1})_{11}&=1+\frac{\pi i c}{\sqrt{\zeta}
}G_-^0(\zeta)\widehat{\psi}(\sqrt{\zeta})
\widehat{\overline{\psi}}(-\sqrt{\zeta}),\label{S11}\\
(S^0(\zeta)^{-1})_{12}&=\frac{\pi i c}{\sqrt{\zeta}
}G_-^0(\zeta)\widehat{\psi}(\sqrt{\zeta})
\widehat{\overline{\psi}}(\sqrt{\zeta}),\label{S12}\\
(S^0(\zeta)^{-1})_{21}&=\frac{\pi i c}{\sqrt{\zeta}
}G_-^0(\zeta)\widehat{\psi}(-\sqrt{\zeta})
\widehat{\overline{\psi}}(-\sqrt{\zeta}),\label{S21}\\
(S^0(\zeta)^{-1})_{22}&=1+\frac{\pi i c}{\sqrt{\zeta}
}G_-^0(\zeta)\widehat{\psi}(-\sqrt{\zeta})
\widehat{\overline{\psi}}(\sqrt{\zeta})\label{S22}.
\end{align}

Next we consider $\eps>0$. Using \eqref{U-def} and \eqref{Teps-def} we have for $z\in\cK=\bC$
\begin{equation}\label{Teps}
T^{\eps}(\lambda;A)z=\frac{1}{\sqrt{2\pi}}\frac{1}{\sqrt{\eps}}
\int_{-\infty}^{\infty} e^{i\lambda p/\eps}
e^{-ip^3/(3\eps)}\widehat{\psi}(p)dp \cdot z.
\end{equation}
We want to continue analytically in $\lambda$ into $\Omega$.

To this end we study the integral in \eqref{Teps}. Assume $u\in\cS(\bR)$ and define
\begin{equation}
\Gamma^{\eps}(\lambda)u=\frac{1}{2\pi\sqrt{\eps}}\int_{-\infty}^{\infty}\Bigl(
\int_{-\infty}^{\infty}
e^{ip(\lambda/\eps)-ip^3/(3\eps)-ixp}u(x)dx\Bigr)dp,
\end{equation}
where the successive integrals converge absolutely.
Thus, it can also be represented 
as the limit of the double integral 
\begin{equation}\label{eq4}
\Gamma^{\eps}(\lambda)u=\lim_{\delta\downarrow0}\frac{1}{2\pi\sqrt{\eps}}
\int_{-\infty}^{\infty}\int_{-\infty}^{\infty} e^{ip(\lambda/\eps)-ip^3/(3\eps)-ixp-\delta p^2}u(x)dxdp.
\end{equation}
We note that $\Gamma^{\eps}(\lambda)u=(U(\eps)u)(\lambda)$. 

From \eqref{eq4} we see that for $z\in\cK=\bC$ 
\begin{equation}\label{eq5}
(\Gamma^{\eps}(\lambda))^{\ast}z=\Bigl(\lim_{\delta\downarrow0}\frac{1}{2\pi\sqrt{\eps}}
\int_{-\infty}^{\infty} e^{-ip(\lambda/\eps)+ip^3/(3\eps)+ixp-\delta p^2}dp\Bigr)z.
\end{equation}
We continue the function inside the parentheses in \eqref{eq5} from $\lambda\in\bR$ to $\zeta\in\bC$ and define
\begin{equation}\label{eq6}
\cG^{\eps}(\zeta,x)=\lim_{\delta\downarrow0}\frac{1}{2\pi\sqrt{\eps}}
\int_{-\infty}^{\infty} e^{-ip(\zeta/\eps)+ip^3/(3\eps)+ixp-\delta p^2}dp.
\end{equation}
For $\zeta\in\bR$ and $x\in\bR$ we have that $\cG^{\eps}(\zeta,x)\in\bR$, since then the imaginary part of the integrand in \eqref{eq6} is an odd function of $p$.

\begin{lemma}\label{lemma1}
We have the following results:
\begin{itemize} 
\item[\rm(1)] The limit in \eqref{eq6} is uniform in compact subsets of $\bC\times\bR$.
\item[\rm(2)] For any $\eta>0$ we can write
\begin{equation}\label{eq312}
\cG^{\eps}(\zeta,x)=\frac{e^{\eta(\zeta/\eps-x)+\eta^3/(3\eps)}}{2\pi\sqrt{\eps}}
\int_{-\infty}^{\infty} e^{-p^2\eta/\eps-i(\zeta p-p^3/3+p\eta^2)/\eps+ixp}dp.
\end{equation}
\item[\rm(3)] $\cG^{\eps}(\zeta,x)$ can be extended to an entire function of $(\zeta,x)\in\bC\times\bC$. For all $(\zeta,x)\in\bC\times\bR$ we have $\overline{\cG^{\eps}(\overline{\zeta},x)}=\cG^{\eps}(\zeta,x)$.
\item[\rm(4)] $\cG^{\eps}(\zeta,x)$ satisfies
\begin{equation}
\bigl(-\frac{d^2}{dx^2}+\eps x -\zeta\bigr)\cG^{\eps}(\zeta,x)=0,\quad
(\zeta,x)\in\bC\times\bR.
\end{equation}
\end{itemize}
\end{lemma}
\begin{proof}
Let $(\zeta,x)\in\bR\times\bR$, $\eps>0$, and $\ell>0$ be fixed. Then there exists a constant $C_0$ such that for $0\leq\eta\leq\ell$ and $p\in\bR$ we have
\begin{align}
\re\bigr(-i(p+i\eta)
(\zeta/\eps)&+i(p+i\eta)^3/(3\eps)+ix(p+i\eta)-\delta(p+i\eta)^2\bigr)\notag\\
&=-p^2\eta/\eps-\delta p^2 + \eta^3/(3\eps)-x\eta+\eta^2\delta+\eta\zeta/\eps\notag\\
&\leq -(\delta+\eta/\eps)p^2+C.\label{eq7}
\end{align}
Then using Cauchy's theorem we can change the integration contour to $\im p=i\eta$ for any $\eta>0$ such that
\begin{equation}
\cG^{\eps}(\zeta,x)=\lim_{\delta\downarrow0}\frac{1}{2\pi\sqrt{\eps}}
\int_{-\infty}^{\infty} e^{-i(p+i\eta)
(\zeta/\eps)+i(p+i\eta)^3/(3\eps)+ix(p+i\eta)-\delta(p+i\eta)^2}dp.\label{eq8}
\end{equation}
Then using \eqref{eq7} we conclude that for any $\eta>0$ the limit in \eqref{eq8} (hence also the limit in \eqref{eq6}) exists  uniformly with respect to $(\zeta,x)$ in a compact subset of $\bR\times\bR$ along with all derivatives. We obtain
\begin{equation*}
\cG^{\eps}(\zeta,x)=\frac{1}{2\pi\sqrt{\eps}}
\int_{-\infty}^{\infty}e^{-i(p+i\eta)(\zeta/\eps)+i(p+i\eta)^3/(3\eps)+ix(p+i\eta)}dp,
\end{equation*}
which may be written in the form \eqref{eq312}.

 It follows that $\cG^{\eps}(\zeta,x)$ can be extended to an entire function of $(\zeta,x)\in\bC\times\bC$. For $x\in\bR$ we have that $\cG^{\eps}(\zeta,x)$ is real for $\zeta\in\bR$. The reflection principle from complex analysis implies that 
$\overline{\cG^{\eps}(\overline{\zeta},x)}=\cG^{\eps}(\zeta,x)$.
We leave the proof of part (4) to the reader.
\end{proof}

We study the behavior of $\cG^{\eps}(\zeta)$ as $\eps\downarrow0$ in the sector $\{\zeta\in\bC\,|\, -\pi/3<\arg \zeta <0\}$. We first represent it using the Airy function. We write $K\Subset\bR$ for a compact subset.
\begin{lemma}
Let $K\Subset\bR$ and let $M\Subset\{\zeta\in\bC\,|\, -\pi/3<\arg \zeta <0\}$.
Then we have for $x\in K$ and $\zeta\in M$
\begin{equation}\label{Ai-G}
\cG^{\eps}(\zeta,x)=\frac{1}{\eps^{\frac16}}\Ai(\omega), 
\quad \omega=\eps^{\frac13}x-\eps^{-\frac23}\zeta,
\end{equation}
where $\Ai(\omega)$ denotes the Airy function
\begin{equation}
\Ai(\omega)=\frac{1}{2\pi i}\int_{\infty e^{-i\pi/3}}^{\infty e^{i\pi/3}}
\exp({t^3}/{3}-\omega t)dt.
\end{equation}
Here the integral is computed over the halflines $\infty e^{-i\pi/3}\to0
\to\infty e^{i\pi/3}$.
\end{lemma}
\begin{proof}
We first make the change of variables 
$q= -ip$ or $p=iq$ in the integral in \eqref{eq6} 
and then $q=\eps^{\frac13}t$ and write it as 
the line integral in the complex plane
\begin{align}
\cG^{\eps}(\zeta,x)&=\lim_{\delta\downarrow0}\frac{1}{2i\pi\sqrt{\eps}}
\int_{-i\infty}^{i\infty}e^{q(\zeta/\eps)+q^3/(3\eps)-xq+\delta q^2}dq\\
&=\lim_{\delta\downarrow0}\frac{1}{2i\pi\eps^{\frac16}}
\int_{-i\infty}^{i\infty}e^{t^3/3-t\omega+\delta\eps^{\frac23}t^2}dt,
\end{align} 
where $\omega= \eps^{\frac13}x-\eps^{-\frac23}\zeta$.

We now want to deform the contour. We note the following implications
\begin{gather*}
-\pi/2 \leq \arg t \leq -\pi/3 \Rightarrow 
-3\pi/2 \leq \arg t^3 \leq -\pi, \quad 
-\pi \leq \arg t^2 \leq -2\pi/3; \\
\pi/3 \leq \arg t \leq \pi/2 \Rightarrow 
\pi \leq \arg t^3 \leq 3\pi/2, \quad 
2\pi/3 \leq \arg t^2 \leq \pi. 
\end{gather*}
Thus for $t \in \{-\pi/2 \leq \arg t \leq -\pi/3\}
\cup \{\pi/3 \leq \arg t \leq \pi/2\}$, we have 
$\re t^3 \leq 0$ and $\re t^2 \leq0$ and we may deform 
the contour of integration to the line graph 
$e^{-\pi/3}\infty \to 0 \to e^{\pi/3}\infty$ 
on which $\arg t^3 = \pi$ or $\arg t^3=-\pi$ 
and $t^3<0$. 
Thus the limit $\delta \to 0$ may be taken inside 
the integral sign and we obtain the desired expression \eqref{Ai-G}.
\end{proof}

Now we define $\rho=-\zeta$ for $\zeta \in M$ so that 
$2\pi/3+\kappa<\arg \rho<\pi-\kappa$ for a $\kappa>0$ and 
\begin{equation*}
\omega= \eps^{\frac13}x-\eps^{-\frac23}\zeta
= \eps^{-\frac23}\rho(1 + \eps (x/\rho)).
\end{equation*}
Then for sufficiently small $\eps_0>0$ there 
exists another constant $\kappa>0$ such that 
for any $0<\eps<\eps_0$, $x\in K$ and $\zeta\in M$, 
$1+\eps(x/\rho)$ is a small perturbation of $1$ 
and $2\pi/3+\kappa < \arg \omega < \pi-\kappa$. It follows 
from (9.5.4) and (9.7.5) in~\cite{DLMF} 
that $\Ai(\omega)$ has the following asymptotic 
expansion as $\abs{\omega}\to \infty$ or $\eps \downarrow 0$
\begin{equation*}
\Ai(\omega) \sim \frac{e^{-\xi}}{2\sqrt{\pi}\omega^{1/4}}
\sum_{k=0}^\infty (-1)^k \frac{u_k}{\xi^k},
\end{equation*}
where $\xi$ is defined in~\cite[(9.7.1)]{DLMF} (the notation there is $\zeta$)
as the principal branch of
\begin{equation*}
\xi=\frac23 \omega^{\frac32}, 
\end{equation*}
and $u_0=1$ and $u_1, \dots$ are constants defined in~\cite[(9.7.2)]{DLMF}.
Note that $\pi <\arg \xi < 3\pi/2$ and 
$0< \arg(-\xi)< \pi/2$ so that $\re(-\xi)>0$ and 
$\Ai(\omega)$ blows up as $\eps \downarrow 0$. 

Using the binomial formula 
$(1+\tau)^{3/2}=1+\frac32 \tau+\frac38 \tau^2+O(\tau^3)$, 
we have  as $\eps \downarrow 0$ uniformly 
with respect to $x\in K$ and $\zeta\in M$ that
\begin{align*}
\xi& = \frac{2}{3}\frac{{\rho}^\frac32}{\eps}
\Bigl(1+\frac{\eps x}{\rho}\Bigr)^{\frac32} 
= \frac{2}{3}\frac{{\rho}^\frac32}{\eps}
\Bigl\{1+ \frac{3\eps}{2}\frac{x}{\rho}
+ \frac{3\eps^2}{8}\Bigl(\frac{x}{\rho}\Bigr)^2 + 
O\Bigl(\frac{\eps x}{\rho}\Bigr)^3\Bigr\} \\
& = \frac{2\rho^\frac32}{3\eps} + x\rho^\frac12 
+ \frac{\eps}{4}\frac{x^2}{\rho^{\frac12}} 
+ O(\eps^2), 
\end{align*} 
hence 
\begin{equation*}
e^{-\xi}= \exp\Bigl(
-\frac{2\rho^\frac32}{3\eps} - x\rho^\frac12
\Bigr)\cdot \Bigl(1- \frac{\eps}{4}\frac{x^2}{\rho^\frac12} 
+ O(\eps^2)\Bigr).
\end{equation*}
Applying the binomial formula  to 
$\omega^{-\frac14}= \eps^{\frac16}\rho^{-\frac14}
(1 + \eps (x/\rho))^{-\frac14}$, we obtain  
\begin{equation*}
\frac{1}{\eps^\frac16\omega^\frac14} 
= \rho^{-\frac14}
\bigl(1 -\frac14 \frac{\eps x}{\rho} + O(\eps^2)\bigr).
\end{equation*}
Combining these products with $u_1=5/72$ we get  
\begin{align}
\cG^\eps(\zeta,x)& = \frac{1}{\eps^\frac16}\Ai(\omega) 
= \frac{1}{2\sqrt{\pi}\rho^\frac14}
\exp\Bigl(
-\frac{2\rho^\frac32}{3\eps} - x\rho^\frac12
\Bigr) \notag \\
& \times \Bigl(
1- \frac{\eps}{4}
\frac{x^2}{\rho^\frac12} 
+ O(\eps^2)
\Bigr)
\Bigl(
1 -\frac{\eps}{4} 
\frac{x}{\rho} 
+ O(\eps^2)
\Bigr) 
\Bigl(
1- \frac{3\eps}{2} 
\frac{u_1}{\rho^\frac32}
+ O(\eps^2)
\Bigr) \notag \\
& = \frac1{2\sqrt{\pi}\rho^\frac14}
\exp\Bigl(
-\frac{2\rho^\frac32}{3\eps} - 
x\rho^\frac12
\Bigr)
\Bigl\{
1-\frac{\eps}{4}
\Bigl(
\frac{x^2}{\rho^\frac12}
+ \frac{x}{\rho}+6\frac{u_1}{\rho^\frac32}
\Bigr) 
+ O(\eps^2)
\Bigr\}. \label{asymp-result}
\end{align}
This leads to the following lemma. 

\begin{lemma} \label{asymp-G}
We have  
\begin{equation} \label{asymp}
\lim_{\eps \downarrow 0} \cG^\eps(\zeta,x)\exp
\Bigl(\frac{2\rho^\frac32}{3\eps}\Bigr) 
= \frac{e^{i\frac{\pi}4}}
{2\sqrt{\pi}\zeta^{\frac14}}
e^{- ix\sqrt{\zeta}},
\end{equation}
uniformly with respect to 
$\zeta\in M\Subset \{\zeta\in \bC\,|\, -\pi/3<\im\zeta<0\}$ 
and $x \in K\Subset \bR$.
\end{lemma} 
\begin{proof}
Due to \eqref{asymp-result} we have 
\eqref{asymp} with the right hand side 
\begin{equation*}
\frac{1}{2\sqrt{\pi}(-\zeta)^\frac14}\exp
\bigl(- x(-\zeta)^\frac12\bigr), 
\end{equation*}
and we only need to fix the branch. We have 
$(-\zeta)^\frac14 = \zeta^\frac14 e^{i\frac{\pi}4}$ 
and $(-\zeta)^\frac12 = i \zeta^\frac12$. Thus the result follows.
\end{proof}

We now have all the results needed to continue $T^{\eps}(\lambda;A)$  and $T^{\eps}(\lambda;B)$ analytically to $\bC$, thus in particular to $\Omega$.
Since $B=cA$, we omit statements for $T^{\eps}(\zeta;B)$ and its adjoint. 
Let $\eps>0$.
We have that
\begin{equation}\label{Tstar}
T^{\eps}(\overline{\zeta};A)^{\ast}=\int_{-\infty}^{\infty}\cG^{\eps}(\zeta,x)
\overline{\psi}(x)dx
\end{equation}
and
\begin{equation}\label{T}
T^{\eps}(\zeta;A)=\int_{-\infty}^{\infty}\cG^{\eps}(\zeta,x)\psi(x)dx.
\end{equation}
The integrals are absolutely convergent due to Assumption~\ref{assump31} and Lemma~\ref{lemma1}. The analytic continuation follows from  Lemma~\ref{lemma1}(3).

Since we have analytic continuations of $T^{\eps}(\lambda;A)$ and $T^{\eps}(\lambda;B)$ for all $\eps\geq0$, the results on continuation of resolvents and scattering matrices, and the results on resonances are available from Section~\ref{sect2}. We will use them in  the next sections to obtain our results.

\section{A result for rank one perturbations}\label{sect4}
We now formulate and prove the main result for rank one perturbations. We need the following well known result, cf.~\cite{Y79}. Recall the definition of $G_{0}^{\eps}(\zeta)$ from \eqref{def-G}.
\begin{lemma}\label{lemma41}
Let $K\Subset\bC^-$. Then $G_{0}^{\eps}(\zeta)^{-1}$ converges strongly to $G_{0}^{0}(\zeta)^{-1}$ as $\eps\downarrow0$, uniformly with respect to $\zeta\in K$.
\end{lemma}

\begin{theorem}\label{rank-one}
Let $\psi$ satisfy Assumption~\ref{assump31}. Assume furthermore that $\psi$ is real-valued.
Let $V=c\ket{\psi}\bra{\psi}$, $c\in\bR$, $c\neq0$. Let $H^{\eps}=H_0^{\eps}+V$, $\eps\geq0$.
Assume that there exists a sequence $\eps_n\downarrow0$ as $n\to\infty$, such that each  $H^{\eps_n}$ has a resonance $\zeta_n$, $-\pi/3<\arg\zeta_n<0$. Assume $\zeta_n\to\zeta_0$ as $n\to\infty$ and 
$-\pi/3<\arg\zeta_0<0$. Then $\zeta_0$ is not a resonance of $H^0$. 
\end{theorem}
\begin{proof}
Let the assumptions in the Theorem be satisfied. It follows from Theorem~\ref{main} that
$(S^{\eps_n}(\zeta_n))^{-1}=0$ for all $n\geq1$, since the scattering matrix is multiplication by a scalar. Thus we have from \eqref{Sinv} that
\begin{equation}\label{S-n}
1+2\pi i T^{\eps_n}(\zeta_n;A) \widetilde{G}_-^{\eps_n}(\zeta_n)
T^{\eps_n}(\overline{\zeta_n};B)^{\ast}=0\quad\text{for all $n\geq1$}.
\end{equation}
Since $\im\zeta_n<0$ we can write $G_-^{\eps_n}(\zeta_n)$ instead of $\widetilde{G}_-^{\eps_n}(\zeta_n)$. We can then use Lemma~\ref{lemma41} to conclude that $G_-^{\eps_n}(\zeta_n)\to G_-^{0}(\zeta_0)$ as $n\to\infty$.
 
Next we look at the limit of $T^{\eps_n}(\zeta_n;A)$ as $n\to\infty$. Let $K=\supp \psi$.
We can determine a set $M\Subset \{\zeta\in \bC\,|\, -\pi/3<\arg\zeta<0\}$ such that $\zeta_n\in M$ for all $n$. We recall from Section~\ref{sect3} the notation $\rho_n=-\zeta_n$. Since $\zeta_n\in M$, we can determine $\kappa>0$ such that for all $n$
we have $\frac23\pi+\kappa<\arg\rho_n<\pi-\kappa$. This implies that there exists $\delta>0$ such that $\re\rho_n^{\frac32}<-\delta$. Thus we have that 
\begin{equation*}
\exp((4\rho_n^{\frac32})/(3\eps_n))\to 0\quad\text{as}\quad n\to\infty.
\end{equation*}
Multiply by $\exp((4\rho_n^{\frac32})/(3\eps_n))$ on both sides in \eqref{S-n} and take the limit $n\to\infty$. Using \eqref{Tstar}, \eqref{T}, Lemma~\ref{asymp-G}, Lemma~\ref{lemma41}, and dominated convergence we get that
\begin{equation}
\widehat{\psi}(\sqrt{\zeta_0})G_-^0(\zeta_0)\widehat{\psi}(\sqrt{\zeta_0})=0,
\end{equation}
since $\psi$ is assumed to be real. Since $G_-^0(\zeta_0)\neq0$ we conclude that $\widehat{\psi}(\sqrt{\zeta_0})=0$.

We now use the formulas \eqref{S11}--\eqref{S22} and the assumption that $\psi$ is real to get 
\begin{equation}
S^0(\zeta_0)^{-1}=
\begin{bmatrix}
1 & 0\\
a & 1
\end{bmatrix},
\end{equation}
where $a=(S^0(\zeta_0)^{-1})_{21}$. This matrix is obviously invertible. Theorem~\ref{main} implies that $\zeta_0$ is not a resonance of $H^0$.
\end{proof}

\section{A result for rank $N$ perturbations}
We outline an extension to a rank $N$ perturbation in this section. We assume that
\begin{equation}\label{rank-N}
V=\sum_{k=1}^Nc_k\ket{\psi_k}\bra{\psi_k}.
\end{equation}
We introduce the following assumption.
\begin{assumption}\label{assumpN}
Let $V$ be given by \eqref{rank-N}. Assume that $c_k\in\bR\setminus\{0\}$, $k=1,\ldots,N$, and $\psi_k\in L^2_{\rm comp}(\bR)$, $k=1,\ldots,N$, linearly independent real functions.
Assume that each $\psi_k$ is an even function. 
\end{assumption}

The factorization $V=B^{\ast}A$ is given with $\cK=\bC^N$ by the operators
\begin{equation}
Af=\begin{bmatrix}
\ip{\psi_1}{f}\\
\vdots\\
\ip{\psi_N}{f}
\end{bmatrix}
\quad\text{and}\quad
Bf=\begin{bmatrix}
c_1\ip{\psi_1}{f}\\
\vdots\\
c_N\ip{\psi_N}{f}.
\end{bmatrix}.
\end{equation}
The operator $Q^{\eps}_0(\zeta)=BR_0^{\eps}(\zeta)A^{\ast}$ is an $N\times N$ matrix with matrix elements
\begin{equation}
Q^{\eps}_0(\zeta)_{k\ell}=c_k\ip{\psi_k}{R_0^{\eps}(\zeta)\psi_{\ell}},\quad k,\ell=1,\ldots N.
\end{equation}

The operator $T^0(\lambda;A)\colon\bC^N\to\bC^2$ is given by the following matrix
\begin{equation}
T^0(\lambda;A)=\frac{1}{\sqrt{2}\lambda^{1/4}}\begin{bmatrix}
\widehat{\psi}_1(\sqrt{\lambda}) & \widehat{\psi}_2(\sqrt{\lambda}) &
\cdots& \widehat{\psi}_N(\sqrt{\lambda})\\
\widehat{\psi}_1(-\sqrt{\lambda}) & \widehat{\psi}_2(-\sqrt{\lambda}) &
\cdots& \widehat{\psi}_N(-\sqrt{\lambda})
\end{bmatrix}.
\end{equation}
The operator $T^0(\overline{\lambda};B)^{\ast}
\colon\bC^2\to\bC^N$ is given by the following matrix
\begin{equation}
T^0(\overline{\lambda};B)^{\ast}=\frac{1}{\sqrt{2}\lambda^{1/4}}\begin{bmatrix}
c_1\widehat{\overline{\psi}}_1(-\sqrt{\lambda}) &
c_1\widehat{\overline{\psi}}_1(\sqrt{\lambda})\\
c_2\widehat{\overline{\psi}}_2(-\sqrt{\lambda}) &
c_2\widehat{\overline{\psi}}_2(\sqrt{\lambda})\\
\vdots & \vdots\\
c_N\widehat{\overline{\psi}}_N(-\sqrt{\lambda}) &
c_N\widehat{\overline{\psi}}_N(\sqrt{\lambda})\\
\end{bmatrix}.
\end{equation}

We introduce a shorthand notation for these two marices. We write
\begin{equation}
T^0(\lambda;A)=\frac{1}{\sqrt{2}\lambda^{1/4}}\begin{bmatrix}
\sfr_1\\ \sfr_2
\end{bmatrix}
\quad\text{and}\quad
T^0(\overline{\lambda};B)^{\ast}=\frac{1}{\sqrt{2}\lambda^{1/4}}\begin{bmatrix}
\sfs_1 & \sfs_2
\end{bmatrix}.
\end{equation}
This leads to the result that 
\begin{equation}
T^0(\lambda;A)G_-^0(\lambda)T^0(\overline{\lambda};B)^{\ast}
=\frac{1}{2\sqrt{\lambda}}
\begin{bmatrix}
\sfr_1\widetilde{G}_-^0(\lambda)\sfs_1 & \sfr_1\widetilde{G}_-^0(\lambda)\sfs_2\\
\sfr_2\widetilde{G}_-^0(\lambda)\sfs_1 & \sfr_2\widetilde{G}_-^0(\lambda)\sfs_2
\end{bmatrix}.
\end{equation}
As in Section~\ref{sect3} we can continue into the lower half plane, such that for $\im\zeta<0$ we get
 from \eqref{Sinv} the expression
\begin{equation}\label{S0-formula}
S^0(\zeta)^{-1}=\begin{bmatrix}
1 & 0\\ 0 & 1\end{bmatrix}
+
\frac{\pi i}{\sqrt{\lambda}}
\begin{bmatrix}
\sfr_1G_-^0(\zeta)\sfs_1 & \sfr_1G_-^0(\zeta)\sfs_2\\
\sfr_2G_-^0(\zeta)s_1 & \sfr_2G_-^0(\zeta)s_2
\end{bmatrix}.
\end{equation}
Note that we write $G_-^0$ instead of $\widetilde{G}_-^0$ since we are not using a continuation.

Now we look at the expression for $S^{\eps}(\zeta)^{-1}$  in the case $\eps>0$. Define
\begin{equation}
\Phi^{\eps}_k(\zeta)=
\int_{-\infty}^{\infty}\cG^{\eps}(\zeta,x)\psi_k(x)dx,\quad k=1,\ldots,N.
\end{equation}
Define the matrices
\begin{equation}
\sfu=\begin{bmatrix}
\Phi^{\eps}_1(\zeta)& \Phi^{\eps}_2(\zeta) &\cdots & \Phi^{\eps}_N(\zeta)
\end{bmatrix}
\quad\text{and}\quad
\sfv=\begin{bmatrix}
c_1\overline{\Phi^{\eps}_1(\overline{\zeta})}\\
c_2\overline{\Phi^{\eps}_2(\overline{\zeta})}\\
\vdots\\
c_N\overline{\Phi^{\eps}_N(\overline{\zeta})}
\end{bmatrix}.
\end{equation}
Then we have for $\im\zeta<0$
\begin{equation}
S^{\eps}(\zeta)^{-1}=1+2\pi i\, \sfu G_-^{\eps}(\zeta)\sfv.
\end{equation}

We can now state the following result.
\begin{theorem}\label{rankN}
Let $V$ satisfy Assumption~\ref{rank-N}. 
Let $H^{\eps}=H_0^{\eps}+V$, $\eps\geq0$.
Assume that there exists a sequence $\eps_n\downarrow0$ as $n\to\infty$, such that each  $H^{\eps_n}$ has a resonance $\zeta_n$, $-\pi/3<\arg\zeta_n<0$. Assume $\zeta_n\to\zeta_0$ as $n\to\infty$ and 
$-\pi/3<\arg\zeta_0<0$. Then $\zeta_0$ is not a resonance of $H^0$. 
\end{theorem}
\begin{proof}
We sketch the main steps in the proof. We have by assumption and Theorem~\ref{main} that $S^{\eps_n}(\zeta_n)^{-1}=0$ for all $n\geq1$. Repeating the convergence argument in the proof of Theorem~\ref{rank-one} we can conclude that 
\begin{equation}\label{eqr1s2}
\sfr_1 G_-^0(\zeta_0)\sfs_2=0.
\end{equation}
Now since $\psi_k$ is assumed to be even, we also have that $\widehat{\psi}_k$ is even.
Furthermore $ \psi_k$ is assumed to be real. Thus we have
\begin{equation}
\widehat{\psi}_k(\sqrt{\zeta_0})=
\widehat{\psi}_k(-\sqrt{\zeta_0})=
\widehat{\overline{\psi}}_k(\sqrt{\zeta_0})=
\widehat{\overline{\psi}}_k(-\sqrt{\zeta_0}), \quad k=1,2,\ldots,N.
\end{equation}
This result implies $\sfr_1=\sfr_2$ and $\sfs_1=\sfs_2$. From \eqref{S0-formula} and \eqref{eqr1s2} we conclude that
\begin{equation}
S^0(\zeta_0)^{-1}=\begin{bmatrix}1 & 0 \\ 0 & 1\end{bmatrix},
\end{equation}
such that by Theorem~\ref{main} $\zeta_0$ is not a resonance of $H^0$.
\end{proof}

\subsection*{Acknowledgements} KY thanks Ira Herbst for asking him about the instability of resonances under Stark perturbations. KY is supported by JSPS grant in aid for scientific research No. 16K05242. AJ acknowledges support from the Danish Council of Independent Research $|$ Natural Sciences, Grant DFF4181-00042. 


\end{document}